%% file: main.tex
\documentclass[envcountsame,envcountsect,final,
fleqn]{llncs}
\usepackage[utf8]{inputenc}
\usepackage[backend=biber, style=alphabetic, maxbibnames=10]{biblatex}
\bibliography{refs.bib}
\RequirePackage{amsmath}

\usepackage{amsthm}
\usepackage{tikz,stmaryrd,amssymb,url,color,relsize,mathtools,enumitem}
\usetikzlibrary{arrows.meta}
\usepackage[ruled,vlined,linesnumbered, procnumbered]{algorithm2e}
\SetAlgorithmName{Game}{game}{List of Games}
\usepackage{xspace}
\usepackage{xcolor}
\usepackage{dsfont}
\usepackage[colorlinks,linkcolor=dblue,urlcolor=dblue, citecolor=dgreen]{hyperref}

\input{commands.tex}

\title{On SDVS Sender Privacy\\ In The Multi-Party Setting}
\author{Jeroen van Wier}
\institute{Interdisciplinary Centre for Security, Reliability and Trust, University of Luxembourg}

\begin{document}

\maketitle

\begin{abstract}
    Strong designated verifier signature schemes rely on sender-privacy to hide the identity of the creator of a signature to all but the intended recipient. This property can be invaluable in, for example, the context of deniability, where the identity of a party should not be deducible from the communication sent during a protocol execution. In this work, we explore the technical definition of sender-privacy and extend it from a 2-party setting to an n-party setting. Afterwards, we show in which cases this extension provides a stronger security and in which cases it does not.
\end{abstract}

\section{Introduction}

Digital signatures have many useful applications in our everyday lives, from message authentication to software updates. In many cases, they provide a publicly verifiable way of proving the authenticity of a message. However, sometimes it is desired to prove authenticity only to the intended receiver, or designated verifier, of a message. Designated verifier signature (\DVS) schemes were constructed for this reason, to allow for the signing of a message in such a way that the receiver would be fully convinced of its authenticity, but to third-party observers, the validity of the signature could be denied. Strong designated verifier signature (\SDVS) schemes are the refinement of this idea, with the additional restraint that no-one but the creator and the designated verifier should be able to deduce from a signature who was the creator. While this concept has been studied extensively and is interpreted intuitively in the same way by many, the technical definitions for the property separating \DVS schemes from \SDVS schemes, known as sender-privacy, vary. In this work we analyze and generalize the definitions in current literature and aim to provide a universally applicable way to define this property, particularly focusing on the $n$-party setting. Furthermore, we prove that our general form of sender-privacy can be achieved by combining weaker forms of sender-privacy with non-transferability or unforgeability. 

\subsection{Related work}

Chaum and van Antwerpen first introduced undeniable signatures in \cite{chaum1989undeniable}, which required interaction between the signer and verifier. In 1996 this requirement was removed by Chaum  \cite{chaum1996private} and by Jakobsson et al. \cite{jakobsson1996designated} separately, who introduced designated verifier signatures. These formal definitions were later refined by Saeednia et al. \cite{saeednia2003efficient}. Rivest et al. introduced ring signatures in \cite{rivest2001leak}, which can be interpreted as \DVS when a ring size of 2 is used, although not \SDVS.

An important step was made when Laguillaumie and Vergnaud formalised \emph{sender-privacy}, the property separating \DVS from \SDVS, in \cite{laguillaumie2004designated}. The notion of \SDVS was further refined to Identity-Based \SDVS by Susilo et al. \cite{susilo2004identity}, where all private keys are issued using a master secret key (i.e. central authority). For this setting, sender-privacy was later formalized in a game-based manner by Huang et al. \cite{huang2006short}.

\section{Preliminaries}

We denote with $\kappa \in \N$ the security parameter of a scheme and implicitly assume that any algorithm that is part of a scheme is given input $1^\kappa$, i.e. the string of $\kappa$ $1$'s, in addition to its specified inputs. We implicitly assume that all adversaries are probabilistic polynomial-time Turing machines (\PPT), although the results also hold if all adversaries are probabilistic polynomial-time quantum Turing machines (\QPT). We write $[n]$ for the set $\{0, \dots, n\}$. We call a function $\eps(n)$ negligible (denoted $\eps \leq \negl(n)$) if for every polynomial $p$ there exists $n_0 \in \N$ such that for all $n \geq n_0$ it holds that $\eps(n) < \frac{1}{p(n)}$. We reserve $\bot$ as an error symbol.

\begin{definition}
    A \emph{designated verifier signature scheme (\DVS scheme)} is a tuple $(\Setup, \KeyGen, \Sign, \Verify, \Simulate)$ of \PPT algorithms such that:
    \begin{itemize}
        \item $\Setup$: Produces the public parameters of a scheme, $\params$. It is implicitly assumed that these parameters are passed to the following algorithms.
        \item $\KeyGen$: Produces a keypair $(\pk, \sk)$.
        \item $\Sign[S][V](m) := \Sign(\sk_S, \pk_S, \pk_V, m)$: Upon input of a sender's keypair, a verifier's public key, and a message $m$, produces a signature $\sigma$ if all keys are valid and $\bot$ otherwise.
        \item $\Verify[S][V](m, \sigma) := \Verify(\sk_V, \pk_V, \pk_S, m, \sigma)$: Upon input of a verifier's keypair, a sender's public key, a message $m$, and a signature $\sigma$, outputs the validity of $\sigma$ (a boolean value) if all keys are valid and $\bot$ otherwise.
        \item $\Simulate[S][V](m) := \Simulate(\sk_V, \pk_V, \pk_S, m)$: Upon input of a verifier's keypair, a sender's public key, and a message $m$, produces a simulated signature $\sigma'$.
    \end{itemize}
\end{definition}

\subsection{Current definitions}
The original definitions for strong verifier designation are a combination of what we currently distinguish as \emph{non-transferability} and \emph{sender privacy}. The following definitions are the initial attempts at defining strong verifier designation, and in their respective papers, they are accompanied by definitions for (non-strong) verifier designation, which are very much in line with the intuition behind non-transferability.

\begin{definition}[\cite{jakobsson1996designated}]
    Let $(\MP_A, \MP_B)$ be a protocol for Alice to prove the truth of the statement $\Omega$ to Bob. We say that Bob is a \emph{\JSI strong designated verifier} if, for any protocol $(\MP_A, \MP_B, \MP_C, \MP_D)$ involving Alice, Bob, Cindy, and Dave, by which Dave proves the truth of some statement $\theta$ to Cindy, there is another protocol $(\MP_C, \MP'_D)$ such that Dave can perform the calculations of $\MP'_D$, and Cindy cannot distinguish transcripts of $(\MP_A, \MP_B, \MP_C, \MP_D)$ from those of $(\MP_C, \MP'_D)$.
\end{definition}

In the above definition, the intuition is that Alice proves a statement to Bob, e.g. the authenticity of a given message. Dave observes this interaction and tries to prove this observation to Cindy. However, strong designation in this sense prevents him from doing so, as any proof he could present to Cindy is indistinguishable (to Cindy) from a simulated proof.

\begin{definition}[\cite{saeednia2003efficient}]
    Let $\MP(A,B)$ be a protocol for Alice to prove the truth of the statement $\Omega$ to Bob. We say that $\MP(A,B)$ is a \emph{\SKM strong designated verifier proof} if anyone can produce identically distributed transcripts that are indistinguishable from those of $\MP(A,B)$ for everybody, except Bob.
\end{definition}

In later work, we see the definition for strong verifier designation split. Non-transferability captures the notion that the verifier can produce signatures from anyone designated to himself, thus ensuring that no signature provides proof of signer-verifier interaction for third parties. Sender privacy adds to this that, from a signature, one cannot deduce the sender, thus allowing no third-party observer to use a signature to plausibly deduce that interaction between two parties happened.
\begin{definition}
    A \DVS scheme $\Pi = (\KeyGen, \Sign, \Verify, \Simulate)$ is \emph{computationally non-transferable} if for any adversary $\MA$,
    \[ \Adv^\NT_{\Pi,\MA}(\kappa) = \Prr[b \in \{0,1\}]{\Gm^\NT_{\Pi,\MA}(\kappa,b) = b} - \frac{1}{2} \leq \negl(\kappa),\]
    where the game $\Gm^\NT_{\Pi, \MA}$ is defined as follows:
    \begin{algorithm}[ht]
    	\caption{$\Gm^\NT_{\Pi, \MA}(\kappa, b)$}
        \DontPrintSemicolon
        $\params \leftarrow \Setup$\;
        $(\pk_S, \sk_S) \leftarrow \KeyGen$\;
        $(\pk_V, \sk_V) \leftarrow \KeyGen$\;
        $(m^*, \state) \leftarrow \MA(1,\params, \pk_S, \sk_S, \pk_V, \sk_V)$\;
        \If{$b=0$}{
            $\sigma^* = \SignL{S}{V}{m^*}$\;
        }
        \Else{
            $\sigma^* = \SimulateL{S}{V}{m^*}$\;
        }
        $b' \leftarrow \MA(2,\state,\sigma^*)$\;
        Output $b'$
    \end{algorithm}
\end{definition}

\begin{definition}
    A \DVS $\Pi = (\KeyGen, \Sign, \Verify, \Simulate)$ is \emph{statistically non-transferable} if for all $S$, $V$, and $m$, $\Sign[S][V](m)$ and $\Simulate[S][V](m)$ are statistically indistinguishable distributions.
\end{definition}

For sender-privacy, many slightly different definitions are presented in literature. Many follow the form of Game \ref{exp:sp}, but with different oracles presented to the adversary. Note that this game is a generalized definition designed to be instantiated with a set of oracles $\MO$ to form the specific definitions found in literature. Besides the oracles, the game takes as parameters the security parameter $\kappa$, the number of parties $n$, and the challenge party index $c$. For each $i\in[n]$, party $i$ is denoted $P_i$. $P_n$ is designated as the verifier for the challenge. In much of the literature this game is played with 3 parties: $S_0$, $S_1$, and $V$, who would here correspond with $P_0$, $P_1$, and $P_2$ respectively in the $n=2$ setting.

\begin{algorithm}[ht]
    \label{exp:sp}
	\caption{$\Gm^\PSI_{\Pi, \MA, \MO}(\kappa, n, c)$, the generalized game for sender-privacy.}
    \DontPrintSemicolon
    $\params \leftarrow \Setup$\;
    $(\pk_{P_0}, \sk_{P_0}) \leftarrow \KeyGen; \dots; (\pk_{P_n}, \sk_{P_n}) \leftarrow \KeyGen$\;
    $(m^*, \state) \leftarrow \MA^{\MO^{(1)}_{sign},\MO^{(1)}_{veri},\MO^{(1)}_{sim}}(1,\params, \pk_{P_0}, \dots, \pk_{P_n})$\;
    $\sigma^* = \Sign[P_c][P_n](m^*)$\;
    $c' \leftarrow \MA^{\MO^{(2)}_{sign},\MO^{(2)}_{veri},\MO^{(2)}_{sim}}(2,\state,\sigma^*)$\;
    Output $c'$
\end{algorithm}


\begin{definition}[\cite{huang2006short}]
    A \DVS $\Pi = (\KeyGen, \Sign, \Verify, \Simulate)$ is a \emph{\Hua-strong} \DVS if it is statistically non-transferable and for any $\PPT$ adversary $\MA$,
    \[ \Adv^{\PSI}_{\Pi,\MA}(\kappa) = \Prr[c \leftarrow \{0,1\}]{\Gm^\PSI_{\Pi,\MA}(\kappa, 2, c) = c} - \frac{1}{2} \leq \negl(\kappa),\]
    where $\Gm^\PSI_{\Pi,\MA}$ is played with the following oracles:
    \begin{itemize}
        \item $\MO^{(1)}_{sign}$: Upon input $(m_i, d_i)$ returns $\Sign[P_{d_i}][P_2](m_i)$ if $d_i \in \{0, 1\}$ and $\bot$ otherwise.
        \item $\MO^{(2)}_{sign}$: Upon input $(m_i, d_i)$ returns $\Sign[P_{d_i}][P_2](m_i)$ if $d_i \in \{0, 1\}$ and $m_i \neq m^*$,and $\bot$ otherwise.
        \item $\MO^{(1)}_{veri}$: Upon input $(\sigma_i, m_i, d_i)$ returns $\Verify[P_{d_i}][P_2](m_i)$ if $d_i \in \{0, 1\}$ and $\bot$ otherwise.
        \item $\MO^{(2)}_{veri}$: Upon input $(\sigma_i, m_i, d_i)$ returns $\Verify[P_{d_i}][P_2](m_i)$ if $d_i \in \{0, 1\}$, $\sigma_i \neq \sigma^*$, and $m_i \neq m^*$,and $\bot$ otherwise.
        \item $\MO^{(1)}_{sim} = \MO^{(2)}_{sim} = \emptyset$
    \end{itemize}
\end{definition}

In \cite{huang2006short}, Huang et al. define signer-privacy for identity-based-\SDVS, a similar type of \DVS where all keypairs are issued by a central authority. Here, they allow signing queries from any party to any party, and the adversary is allowed to choose the two signer and the verifier parties. We explore this option for \SDVS in Definition \ref{def:advpsi}.

\section{Bringing sender-privacy to the multi-party setting}
\label{sec:multipartySP}

Sender privacy is meant to provide security in the setting where an eavesdropping adversary is trying to detect the identity of the sender of a signature. In the previously presented definitions, this is modeled by a coin flip between two senders, with a fixed verifier. This way of defining sender privacy is similar to key-privacy in public-key cryptography \cite{bellare2001key}. The key difference here is that public-key ciphertexts are only related to one keypair, the receiver's. However, designated verifier signatures are bound to two parties, the signer and the designated verifier. This creates the problem that the naive way of defining sender-privacy does not cover any attacks that require multiple parties. In key-privacy, any adversary requiring $n$ parties for their attack can perform this attack in the two-party setting by simulating the other $n-2$ parties themself. However, in the case of \SDVS schemes, this is not necessarily possible. The adversary could be unable to create signatures signed by one of the two challenge parties with their simulated parties as the verifier, as is depicted in Figure \ref{fig:npartyadv}. In particular if one does not have statistical non-transferability, this might pose a problem. For this reason, we explicitly shape our definition for the multi-party setting. We explore settings where this is a non-issue in Section \ref{sec:altsettings}.


\tikzstyle{Party}=[fill=white, draw=black, shape=circle, inner sep=1pt]

\tikzstyle{Challenger area}=[-, fill=none, draw=black, dashed=true]
\tikzstyle{Adversary area}=[-, fill=none, draw=red, dashed=true]
\tikzstyle{Sign}=[arrows={->[scale=2]}, fill=none, draw=black]
\tikzstyle{NoSign}=[arrows={->[scale=2,red]}, fill=none, draw=red]

\begin{figure}[ht]
    \centering
        \begin{tikzpicture}
        		\node [style=Party] (0) at (-3.25, 14.5) {$P_4$};
        		\node [style=Party] (1) at (-1.25, 14.5) {$P_5$};
        		\node [style=Party] (2) at (-3.25, 13) {$P_2$};
        		\node [style=Party] (3) at (-1.25, 13) {$P_3$};
        		\node [style=Party] (4) at (-3.25, 11.5) {$P_0$};
        		\node [style=Party] (5) at (-1.25, 11.5) {$P_1$};
        		\node [label=220:$\CC$] (6) at (-3.75, 15) {};
        		\node (7) at (-3.75, 11) {};
        		\node (8) at (-0.75, 11) {};
        		\node (9) at (-0.75, 15) {};
        		\draw [style=Challenger area] (6.center) to (9.center);
        		\draw [style=Challenger area] (9.center) to (8.center);
        		\draw [style=Challenger area] (8.center) to (7.center);
        		\draw [style=Challenger area] (7.center) to (6.center);
        		\draw [style=Sign] (3.north) to (1.south);
        		\node (20) at (-2, 13.75) {$\MO_{sign}$};
        \end{tikzpicture}\qquad\qquad\qquad
        \begin{tikzpicture}
                \node [style=Party] (0) at (-3.25, 14.5) {$P_4$};
        		\node [style=Party] (1) at (-1.25, 14.5) {$P_5$};
        		\node [style=Party] (2) at (-3.25, 13) {$P_2$};
        		\node [style=Party] (3) at (-1.25, 13) {$P_3$};
        		\node [style=Party] (4) at (-3.25, 11.5) {$P_0$};
        		\node [style=Party] (5) at (-1.25, 11.5) {$P_1$};
        		\node [label=220:{\color{red}$\MA$}] (6) at (-3.75, 15) {};
        		\node (7) at (-3.75, 14) {};
        		\node (8) at (-0.75, 14) {};
        		\node (9) at (-0.75, 15) {};
        		\node [label=220:{$\CC$}] (10) at (-3.75, 13.5) {};
        		\node (11) at (-3.75, 11) {};
        		\node (12) at (-0.75, 11) {};
        		\node (13) at (-0.75, 13.5) {};
        		\draw [style=Adversary area] (6.center) to (9.center);
        		\draw [style=Adversary area] (9.center) to (8.center);
        		\draw [style=Adversary area] (8.center) to (7.center);
        		\draw [style=Adversary area] (7.center) to (6.center);
        		\draw [style=Challenger area] (10.center) to (13.center);
        		\draw [style=Challenger area] (13.center) to (12.center);
        		\draw [style=Challenger area] (12.center) to (11.center);
        		\draw [style=Challenger area] (11.center) to (10.center);
        		\draw [style=NoSign] (3.north) to (1.south);
        
        \end{tikzpicture}
    \caption{Left: a 6-party setting where the adversary requests a signature using an oracle, Right: a 4-party setting where the adversary simulates another 2 parties but is now unable to obtain the same signature as on the left.}
    \label{fig:npartyadv}
\end{figure}
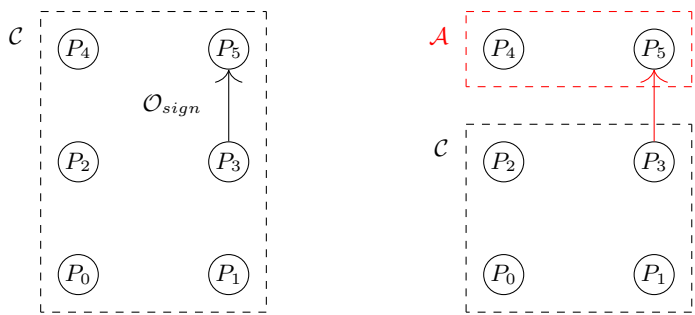

\subsection{Oracles}

Many different interpretations exist in the literature of what oracles the adversary should be given access to. The key choices here are whether (1) a simulation oracle should be provided, (2) a verification oracle should be provided, and (3) whether the adversary should still have access to the oracles after the challenge has been issued. Whereas the precise attacker model might depend on the context and our framework allows us to capture this, we here choose to focus on the strongest level of security, by providing the adversary with as much as possible without trivially breaking the challenge.

\begin{definition}
    For any $n$, let the \emph{standard $n$-sender \PSI-oracles} denote:
    \begin{itemize}
        \item $\MO^{(1)}_{sign} = \MO^{(2)}_{sign}$: Upon input $(m_i, s, v)$ returns $\sigma_i := \Sign[P_s][P_v](m_i)$ if $s, v \in [n]$ and $\bot$ otherwise.
        \item $\MO^{(1)}_{sim} = \MO^{(2)}_{sim}$: Upon input $(m_i, s, v)$ returns $\sigma_i := \Simulate[P_s][P_v](m_i)$ if $s, v \in [n]$ and $\bot$ otherwise.
        \item $\MO^{(1)}_{veri}$: Upon input $(m_i, \sigma_i, s, v)$ returns $\Verify[P_s][P_v](m_i, \sigma_i)$ if $s, v \in  [n]$ and $\bot$ otherwise.
        \item $\MO^{(2)}_{veri}$: Upon input $(m_i, \sigma_i, s, v)$ returns $\Verify[P_s][P_v](m_i, \sigma_i)$ if $s, v \in  [n]$ and $\sigma_i \neq \sigma^*$, and $\bot$ otherwise.
    \end{itemize}
\end{definition}

Note that the oracles make use of an implicit ordering of the parties. This makes no difference in any real-world application, but for constructing proofs we also define a set of oracles that allows this ordering to be hidden by a permutation.

\begin{definition}
    For any set of oracles 
    for $\Gm^\PSI$ and any permutation $\pi$ define the permuted oracles as follows, where $b\in\{0,1\}$:
    \begin{itemize}
        \item $\MO^{(\pi,b)}_{sign}$: On input $(m_i,s,v)$ output $\MO^{(b)}_{sign}(m_i, \pi(s), \pi(v))$
        \item $\MO^{(\pi,b)}_{sim}$: On input $(m_i,s,v)$ output $\MO^{(b)}_{sim}(m_i, \pi(s), \pi(v))$
        \item $\MO^{(\pi,b)}_{veri}$: On input $(m_i, \sigma_i,s,v)$ output $\MO^{(b)}_{veri}(m_i, \sigma_i, \pi(s), \pi(v))$
    \end{itemize}
\end{definition}

\subsection{Definition}

Taking all these things into consideration, we can now craft a definition of sender privacy. This definition is more in line with current research in ID-based-\SDVS research such as \cite{huang2011identity}.

\begin{definition}\label{def:psi}
    A \DVS scheme $\Pi$ is \emph{$n$-party sender private with respect to $\MO$} if for any adversary $\MA$,
    \[ \Adv^{\PSI}_{\Pi,\MA,\MO}(\kappa, n) = \Prr[c \leftarrow \{0,1\}]{\Gm^\PSI_{\Pi,\MA}(\kappa, n, c) = c} - \frac{1}{2} \leq \negl(\kappa).\]
    A \DVS scheme is \emph{$n$-party sender private} if it is $n$-party sender private with respect to the standard $n$-sender \PSI-oracles.
\end{definition}

\section{Alternative definitions}

In this section, we look at possible alternative definitions that one could consider equally valid generalizations of the 2-party setting to the $n$-party setting. For example, in the 2-party setting, we pick the challenge uniformly at random between the two possible senders, thus one could consider picking uniformly at random from $n$ senders in the $n$-party setting.


\begin{definition}
    A \DVS scheme is \emph{$n$-party random-challenge sender private with respect to $\MO$} if for any adversary $\MA$,
    \[ \Adv^{nr\PSI}_{\Pi,\MA, \MO}(\kappa, n) = \Prr[c \leftarrow {[n-1]}]{\Gm^\PSI_{\Pi,\MA,\MO}(\kappa, n, c) = c} - \frac{1}{n} \leq \negl(\kappa).\]
    A \DVS scheme is \emph{$n$-party random-challenge sender private} if it is $n$-party random-challenge sender private with respect to the standard n-sender \PSI-oracles.
\end{definition}


Furthermore, one could strengthen the definition even more by allowing the adversary to choose which two senders the challenge is chosen from and which party is the verifier. 

\begin{algorithm}[ht]
    \label{exp:advsp}
	\caption{$\Gm^\ADVPSI_{\Pi, \MA, \MO}(\kappa, n, c)$}
    \DontPrintSemicolon
    $\params \leftarrow \Setup$\;
    $(\pk_{P_0}, \sk_{P_0}) \leftarrow \KeyGen; \dots; (\pk_{P_n}, \sk_{P_n}) \leftarrow \KeyGen$\;
    $(m^*, s_0, s_1, r, \state) \leftarrow \MA^{\MO^{(1)}_{sign},\MO^{(1)}_{veri},\MO^{(1)}_{sim}}(1,\params, \pk_{P_0}, \dots, \pk_{P_n})$\;
    $\sigma^* = \Sign[P_{s_c}][P_r](m^*)$\;
    $c' \leftarrow \MA^{\MO^{(2)}_{sign},\MO^{(2)}_{veri},\MO^{(2)}_{sim}}(2,\state,\sigma^*)$\;
    Output $c'$
\end{algorithm}

\begin{definition}
    \label{def:advpsi}
    A \DVS scheme is \emph{$n$-party adversarial-challenge sender private with respect to $\MO$} if for any adversary $\MA$,
    \[ \Adv^\ADVPSI_{\Pi,\MA, \MO}(\kappa, n) = \Prr[c \leftarrow \{0,1\}]{\Gm^\ADVPSI_{\Pi,\MA}(\kappa, n, c) = c} - \frac{1}{2} \leq \negl(\kappa).\]
    A \DVS scheme is \emph{$n$-party adversarial-challenge sender private} if it is $n$-party adversarial-challenge sender private with respect to the standard $n$-sender \PSI-oracles.
\end{definition}

\subsection{Relations}

As one might expect, the above-defined alternative definitions relate strongly to the main definition, Definition \ref{def:psi}. In fact, in this section, we show that they are equivalent up to polynomial differences in the advantages.

For the universally random challenge, this can be done by simply only considering the cases where the challenge is $P_0$ or $P_1$, which will be the case 2 out of $n$ times, giving us a loss in the advantage of a factor $\frac{2}{n}$.

\begin{theorem}\label{thm:nr2nf}
    For any adversary $\MA$, \DVS scheme $\Pi$, and set of oracles $\MO$, 
    \[ \frac{2}{n} \cdot \Adv^{\PSI}_{\Pi,\MA,\MO}(\kappa, n) \leq \Adv^{nr\PSI}_{\Pi,\MA,\MO}(\kappa, n)\]
\end{theorem}
\begin{proof}
    \begin{align*}
        &\hspace{-5mm}\Adv^{nr\PSI}_{\Pi,\MA}(\kappa,n)\\
        \quad&= \Prr[c \leftarrow {[n-1]}]{\Gm^\PSI_{\Pi,\MA}(\kappa, n, c) = c} - \frac{1}{n}\\
        &= \frac{2}{n}\Prr[c \leftarrow {[1]}]{\Gm^\PSI_{\Pi,\MA}(\kappa, n, c) = c} + \frac{n-2}{n}\Prr[{c \leftarrow [2,n-1]}]{\Gm^\PSI_{\Pi,\MA}(\kappa, n, c) = c} - \frac{1}{n}\\
        &= \frac{2}{n}\left(\Prr[c \leftarrow {[1]}]{\Gm^\PSI_{\Pi,\MA}(\kappa, n, c) = c} - \frac{1}{2}\right) + \frac{n-2}{n}\Prr[{c \leftarrow [2,n-1]}]{\Gm^\PSI_{\Pi,\MA}(\kappa, n, c) = c}\\
        &= \frac{2}{n}\cdot \Adv^{\PSI}_{\Pi,\MA}(\kappa) + \frac{n-2}{n}\Prr[c \leftarrow \{2,\dots,n-1\}]{\Gm^\PSI_{\Pi,\MA}(\kappa, n, c) = c}\\
        &\geq \frac{2}{n}\cdot \Adv^{\PSI}_{\Pi,\MA, \MO}(\kappa, n),
    \end{align*}
    where $[2,n-1] = \{2, \dots, n-1\}$.
\end{proof}

\begin{theorem}\label{thm:nf2nr}
    For any adversary $\MA$, set of oracles $\MO$ and \DVS scheme $\Pi$, there exists an adversary $\MB$ such that
    \[ \frac{1}{2}\Adv^{nr\PSI}_{\Pi,\MA,\MO}(\kappa,n) \leq \Adv^{\PSI}_{\Pi,\MB,\MO}(\kappa,n).\]
\end{theorem}
\begin{proof}
    Here, we omit the subscripts $\Pi$ and $\MO$ for $\Adv$ and $\Gm$ for simplicity. Let $\MB$ be defined as in Games \ref{gm:nf2nrb1} and \ref{gm:nf2nrb2}.
    \begin{algorithm}[ht]\label{gm:nf2nrb1}
	\caption{$\MB^{\MO^{(1)}_{sign},\MO^{(1)}_{veri},\MO^{(1)}_{sim}}(1,\params, \pk_{P_0}, \dots, \pk_{P_n})$}
    \DontPrintSemicolon
    Pick a random permutation $\pi: [n] \mapsto [n]$ such that $\pi(n) = n$\;
    $(m^*, \state) \leftarrow\MA^{\MO^{(\pi,1)}_{sign},\MO^{(\pi,1)}_{veri},\MO^{(\pi,1)}_{sim}}(1,\params, \pk_{P_{\pi(0)}}, \dots, \pk_{P_{\pi(n)}})$\;
    Output $(m^*, (\pi, \state))$
    \end{algorithm}
    \begin{algorithm}[ht]\label{gm:nf2nrb2}
	\caption{$\MB^{\MO^{(2)}_{sign},\MO^{(2)}_{veri},\MO^{(2)}_{sim}}(2,\state', \sigma^*)$}
    \DontPrintSemicolon
    Parse $\state'$ as $(\pi,\state)$\;
    $c' \leftarrow \MA^{\MO^{(\pi,2)}_{sign},\MO^{(\pi,2)}_{veri},\MO^{(\pi,2)}_{sim}}(2,\state,\sigma^*)$\;
    \If{$\pi(c')\in\{0,1\}$}{Output $\pi(c')$}
    \Else{Output $0$}
    \end{algorithm}
    The permutation is used here to hide the indexation of the parties from the adversary. Note that applying a permutation $\pi$ in this fashion is equivalent to  generating the keypairs in the order $\pi^{-1}(0)\dots\pi^{-1}(n)$ and since these are i.i.d. samples the order of their generation does not affect the winning probability of $\MA$. However, it guarantees that the winning probability of $\MA$ is the same for every $c$. Note that here we use $\Pr_{\pi}$ to indicate the uniform probability over all $\pi: [n] \mapsto [n]$ such that $\pi(n)=n$.
    \begin{align*}
        &\hspace{-5mm}\Adv^{\PSI}_{\MB}(\kappa,n)\\ &= \Prr[c \leftarrow {[1]}]{\Gm^\PSI_{\MB}(\kappa, n, c) = c} - \frac{1}{2}\\
        &= \Prr[c \leftarrow {[1]},\pi]{\Gm^\PSI_{\MA}(\kappa, n, \pi^{-1}(c)) = \pi^{-1}(c)} \\&\qquad+ \frac{1}{2}\Prr[\pi]{\Gm^\PSI_{\MA}(\kappa, n, \pi^{-1}(0)) \not\in \{\pi^{-1}(0), \pi^{-1}(1)\}} - \frac{1}{2}\\
        &= \Prr[c \leftarrow {[n-1]}]{\Gm^\PSI_{\MA}(\kappa, n, c) = c} - \frac{1}{2}\Prr[\pi]{\Gm^\PSI_{\MA}(\kappa, n, \pi^{-1}(0)) \in \{\pi^{-1}(0), \pi^{-1}(1)\}}\\
        &= \frac{1}{2}\Prr[c \leftarrow {[n-1]}]{\Gm^\PSI_{\MA}(\kappa, n, c) = c} - \frac{1}{2}\Prr[\pi]{\Gm^\PSI_{\MA}(\kappa, n, \pi^{-1}(0)) = \pi^{-1}(1)}\\
        &= \frac{1}{2}\left(\Prr[c \leftarrow {[n-1]}]{\Gm^\PSI_{\MA}(\kappa, n, c) = c} - \frac{1}{n-1}\Prr[c \leftarrow {[n-1]}]{\Gm^\PSI_{\MA}(\kappa, n, c) \neq c}\right)\\
        &= \frac{1}{2}\left(\frac{n}{n-1}\Prr[c \leftarrow {[n-1]}]{\Gm^\PSI_{\MA}(\kappa, n, c) = c} - \frac{1}{n-1}\right)\\
        &=\frac{n}{2(n-1)}\Adv^{nr\PSI}_{\MA}(\kappa,n) \geq \frac{1}{2}\Adv^{nr\PSI}_{\MA}(\kappa,n)
    \end{align*}
\end{proof}

Combining Theorem \ref{thm:nr2nf} and Theorem \ref{thm:nf2nr}, we see that the advantages for fixed-challenge and random-challenge only differ by at most a linear factor. Thus these definitions are equivalent when considering negligible advantages.

\begin{corollary}
    For any $n \in \N$, an \SDVS scheme is $n$-party random-challenge sender private if and only if it is $n$-party sender private.
\end{corollary}

For adversarially-chosen challenges, we could try to simply consider only the cases where the adversary chooses $P_0$ and $P_1$ as the challenge senders and $P_n$ as the challenge verifier. However, an adversary could be crafted to never choose this exact combination of parties. Thus, we hide the indexation of the parties under a random permutation. This is done only for the proof and has no impact on the actual definition, as all parties' keypairs are i.i.d. samples. Since the adversary does not know this permutation, the chance of them picking these parties is in the order of $n^{-3}$ and thus a loss of this order is incurred in the advantage.

\begin{theorem}\label{thm:nf2adv}
    For any adversary $\MA$ and set of oracles $\MO$, there exists an adversary $\MB$ such that
    \[ \frac{2}{n^3 - n} \cdot \Adv^{\ADVPSI}_{\Pi,\MA, \MO}(\kappa,n) \leq \Adv^{\PSI}_{\Pi,\MB, \MO}(\kappa,n)\]
\end{theorem}
\begin{proof}
    Fix $\MA$. Let $\MB$ be defined as:
    \begin{algorithm}[ht]
	\caption{$\MB^{\MO^{(1)}_{sign},\MO^{(1)}_{veri},\MO^{(1)}_{sim}}(1,\params, \pk_{P_0}, \dots, \pk_{P_n})$}
    \DontPrintSemicolon
    Pick a random permutation $\pi: [n] \mapsto [n]$\;
    $(m^*, s_0, s_1, r, \state) \leftarrow\MA^{\MO^{(\pi,1)}_{sign},\MO^{(\pi,1)}_{veri},\MO^{(\pi,1)}_{sim}}(1,\params, \pk_{P_{\pi(0)}}, \dots, \pk_{P_{\pi(n)}})$\;
    \If{
        $\pi(s_0) = 0 \land \pi(s_1) = 1 \land \pi(r) = n$
    }{
        Output $(m^*, (0,\state))$\;
    }
    \ElseIf{
        $\pi(s_0) = 1 \land \pi(s_1) = 0 \land \pi(r) = n$
    }{
        Output $(m^*, (1,\state))$\;
    }
    \Else{
        Output $(m^*, (2,\state))$\;
    }
    \end{algorithm}
    \begin{algorithm}[ht]
	\caption{$\MB^{\MO^{(2)}_{sign},\MO^{(2)}_{veri},\MO^{(2)}_{sim}}(2,\state', \sigma^*)$}
    \DontPrintSemicolon
    Parse $\state'$ as $(b,\state)$\;
    $c' \leftarrow \MA^{\MO^{(2)}_{sign},\MO^{(2)}_{veri},\MO^{(2)}_{sim}}(2,\state,\sigma^*)$\;
    \If{$b=0$}{
        Output $c'$\;
    }
    \ElseIf{$b=1$}{
        Output $1-c'$\;
    }
    \Else{
        $c'' \leftarrow \{0,1\}$\;
        Output $c''$\;
    }
    \end{algorithm}
    
    The permutation is used here to hide the indexation of the parties from the adversary. Note that applying a permutation $\pi$ in this fashion is equivalent to  generating the keypairs in the order $\pi^{-1}(0)\dots\pi^{-1}(n)$ and since these are i.i.d. samples the order of their generation does not affect the winning probability of $\MA$. When playing game $\Gm^\PSI_{\Pi,\MB}$, we can now distinguish two cases:
    \begin{enumerate}
        \item $\{\pi(s_0), \pi(s_1)\} = \{0,1\}$ and $\pi(r) = n$. Since $\pi$ is random and unknown to $\MA$, this happens with probability $\frac{2(n-2)!}{(n+1)!}$. In this case, $\MA$ has chosen $P_0$ and $P_1$ as the possible signers and $P_n$ as the verifier, making $\Gm^\ADVPSI_{\Pi, \MA}$ and $\Gm^{\PSI}_{\Pi, \MB}$ equivalent.
        \item Otherwise, $\MA$ has chosen different signers or verifiers, in which case $\Gm^{\PSI}_{\Pi, \MB}$ becomes equivalent to a random coin flip, with probability $\frac{1}{2}$ of guessing $c$.
    \end{enumerate}
    Combining this, we get that
    \begin{align*}
        &\Prr[c \leftarrow \{0,1\}]{\Gm^{\PSI}_{\Pi, \MB,\MO}(\kappa, n, c) = c} =\\
        &\qquad\qquad\frac{2(n-2)!}{(n+1)!} \Prr[c \leftarrow \{0,1\}]{\Gm^{\ADVPSI}_{\Pi, \MA, \MO}(\kappa, n, c) = c} + \left(1 - \frac{2(n-2)!}{(n+1)!}\right)\frac{1}{2}.
    \end{align*}
    Thus,
     \[ \Adv^{\PSI}_{\Pi,\MB,\MO}(\kappa,n) = \frac{2(n-2)!}{(n+1)!} \cdot \Adv^{\ADVPSI}_{\Pi,\MA,\MO}(\kappa,n).\]
\end{proof}

\begin{corollary}
    For any $n \in \N$, an \SDVS scheme is $n$-party adversarial-challenge sender private if and only if it is $n$-party sender private.
\end{corollary}
\begin{proof}
    Theorem \ref{thm:nf2adv} shows that if a scheme is sender private then it is also adversarial-challenge sender private since the advantage differs by a factor $\MO(n^3)$. The other direction is trivial, as any adversary for sender-privacy can trivially be transformed into an adversary for adversarial-challenge sender privacy, always outputting $s_0 = 0, s_1 = 1, r = n$, which gives both adversaries the exact same winning probability.
\end{proof}

\section{Alternative oracles}\label{sec:altsettings}

In this section we show that one can use other properties of \SDVS schemes, e.g. non-transferability and unforgeability, to provide equally strong sender-privacy while giving the adversary weaker oracles. This allows us to more easily prove that existing schemes satisfy our definition. Note that in this section we only consider the cases where the security advantages are negligible. First, we will focus on the verification oracle, showing that they can be removed without impacting the quality of the security when the scheme is unforgeable. Then, we show that the number of parties can be limited to 3 ($n=2$) when a scheme is both unforgeable and non-transferable.

\begin{definition}
    A \DVS scheme $\Pi = (\KeyGen, \Sign, \Verify, \Simulate)$ is \emph{$n$-party strongly-unforgeable with respect to $\MO$} if for any adversary $\MA$,
    \[ \Adv^\UF_{\Pi,\MA,\MO}(\kappa,n) = \Prr{\Gm^\UF_{\Pi,\MA,\MO}(\kappa,n) = \top} \leq \negl(\kappa),\]
    where the game $\Gm^\UF_{\Pi, \MA, \MO}$ is defined in Game \ref{gm:su}.
    \begin{algorithm}[h]\label{gm:su}
    	\caption{$\Gm^\UF_{\Pi, \MA, \MO}(\kappa, n)$}
        \DontPrintSemicolon
        $\params \leftarrow \Setup$\;
        $(\pk_{P_0}, \sk_{P_0}) \leftarrow \KeyGen; \dots; (\pk_{P_n}, \sk_{P_n}) \leftarrow \KeyGen$\;
        $(m^*, \sigma^*,s,v) \leftarrow \MA^{\MO_{sign},\MO_{veri},\MO_{sim}}(\params, \pk_{P_0}, \dots, \pk_{P_n})$\;
        \If{$\Verify[P_s][P_v](m^*, \sigma^*) = 1$ and $\forall i: \sigma^* \neq \sigma_i$}{Output $\top$.}
        \Else{Output $\bot$.}
    \end{algorithm}
    A \DVS scheme is \emph{$n$-party strongly-unforgeable} if it is $n$-party strongly-unforgeable with respect to $\MO^{(1)}_{sign}$, $\MO^{(1)}_{sim}$, $\MO^{(1)}_{veri}$ from the standard $n$-sender \PSI-oracles.
\end{definition}


\begin{theorem}\label{thm:uf2veri}
    Let $n\in \N$ and $\MO = \{\MO^{(1)}_{sign}, \MO^{(2)}_{sign}, \MO^{(1)}_{sim}, \MO^{(2)}_{sim}, \MO^{(1)}_{veri}, \MO^{(2)}_{veri}\}$ be the $n$-sender standard oracles. Any \DVS scheme that is $n$-party sender private with respect to $\MO' = \{\MO^{(1)}_{sign}, \MO^{(2)}_{sign}, \MO^{(1)}_{sim}, \MO^{(2)}_{sim}, \MO^{'(1)}_{veri} = \emptyset, \MO^{'(2)}_{veri} = \emptyset\}$ and strongly unforgeable is $n$-party sender private (with respect to $\MO$).
\end{theorem}
\begin{proof}
    Fix $n \in \N$. Suppose \DVS scheme $\Pi$ is $n$-party sender private with respect to $\MO' = \{\MO^{(1)}_{sign}, \MO^{(2)}_{sign}, \MO^{(1)}_{sim}, \MO^{(2)}_{sim}, \MO^{'(1)}_{veri} = \emptyset, \MO^{'(2)}_{veri} = \emptyset\}$ and strongly unforgeable, but not $n$-party sender private with respect to $\MO$. Then there exists an adversary $\MA$ such that $\Adv^\PSI_{\Pi, \MA, \MO}(\kappa) \not\leq \negl(\kappa)$. Let $\MA'$ be $\MA$, except every query $\MO^{(b)}_{veri}(m_i, \sigma_i, s, v)$ is replaced with $\top$ if $(m_i, \sigma_i)$ was the result of a signing or simulating oracle query and $\bot$ otherwise. Since $\MA'$ no longer uses the verification oracles, we have $\Adv^\PSI_{\Pi, \MA', \MO} = \Adv^\PSI_{\Pi, \MA', \MO'} \leq \negl(\kappa)$, i.e. $\MA'$ has the same advantage with respect to $\MO$ and $\MO'$, as they only differ in the verification oracles.
    
    Now consider the adversary $\MB$, who intends to create a forged signature. $\MB$ runs $\MA$, recording all signing and simulating queries. Whenever $\MA$ makes a verification query for a valid signature that was not the result of a signing or simulating query, $\MB$ outputs this signature and halts. Note that the only difference in the behavior of $\MA$ and $\MA'$ can occur when $\MA$ makes such a query. Since the difference between $\Adv^\PSI_{\Pi, \MA', \MO}$ and $\Adv^\PSI_{\Pi, \MA, \MO}$ is more than negligible, we have that such a query occurs with more than negligible probability, giving $\MB$ a more than negligible probability of constructing a forgery. This contradicts the fact that $\Pi$ is strongly unforgeable.
\end{proof}

\begin{theorem}
    Any \DVS scheme $\Pi$ that is $2$-party sender private, strongly unforgeable, and computationally non-transferable is $n$-party sender private for any $n \geq 2$.
\end{theorem}
\begin{proof}
    Suppose a \DVS scheme $\Pi$ is $2$-party sender private, strongly unforgeable, and computationally non-transferable. Assume towards a contradiction that $\Pi$ is not $n$-party sender private for some fixed $n > 2$. By Theorem \ref{thm:uf2veri}, this means $\Pi$ is also not $n$-party sender private with respect to 
    \[\MO' = \{\MO^{(1)}_{sign}, \MO^{(2)}_{sign}, \MO^{(1)}_{sim}, \MO^{(2)}_{sim}, \MO^{'(1)}_{veri} = \emptyset, \MO^{'(2)}_{veri} = \emptyset\}.\]
    Thus, there exists and adversary $\MA$ such that $\Adv^\PSI_{\Pi, \MA, \MO'}(\kappa, n) \not\leq \negl(\kappa)$. Let $\MA'(1,\params, \pk_{P_0}, \pk_{P_1}, \pk_{P_2})$ be as follows: First, sample $n-2$ keypairs $(\sk'_{P_2}, \pk'_{P_2})$ \dots $(\sk'_{P_{n-1}}, \pk'_{P_{n-1}})$ representing parties $P'_2 \dots P'_{n-1}$ and set $P'_0 = P_0$, $P'_1 = P_1$, $P'_n = P_2$. Then, run $\MA$ with the oracles $\MO''$ defined as follows, with $b = 1,2$:
    \begin{itemize}
        \item $\MO^{''(b)}_{veri} = \emptyset$.
        \item $\MO^{''(b)}_{sign}(m_i, s, v):$
        \begin{itemize}
            \item If $s,v \in \{0,1,n\}$, return $\MO^{(b)}_{sign}(m_i, \max(2, s), \max(2,v))$.
            \item If $s \in \{2, \dots, n-1\}$ and $v\in [n]$, return $\Sign[P'_s][P'_v](m_i)$.
            \item If $s \in \{0,1,n\}$ and $v \in \{2, \dots, n-1\}$, return  $\Simulate[P'_s][P'_v](m_i)$.
            \item Else, return $\bot$.
        \end{itemize}
        \item $\MO^{''(b)}_{sim}(m_i, s, v):$
        \begin{itemize}
            \item If $s,v \in \{0,1,n\}$, return $\MO^{(b)}_{sim}(m_i, \max(2, s), \max(2,v))$.
            \item If $v \in \{2, \dots, n-1\}$ and $s\in [n]$, return $\Simulate[P'_s][P'_v](m_i)$.
            \item If $v \in \{0,1,n\}$ and $s \in \{2, \dots, n-1\}$, return  $\Sign[P'_s][P'_v](m_i)$.
            \item Else, return $\bot$.
        \end{itemize}
    \end{itemize}
    
    Note that these oracles make use of the fact that one can simulate or sign a signature without, respectively, the sender's or verifier's secret key. Thus we circumvent the issue mentioned in Section \ref{sec:multipartySP}. In the oracles, $\max$ is used here to map $n$ to $2$, as $n$ and $2$ are the challenge verifiers in the $n$- and $2$-party respectively.
    
    Since $\Pi$ is 2-party sender private, we have $\Adv^\PSI_{\Pi, \MA', \MO''}(\kappa, 2) \leq \negl(\kappa)$. When we replace all oracle calls by their respective functionality, then $\Gm^\PSI_{\Pi, \MA, \MO'}(\kappa, n, c)$ and $\Gm^\PSI_{\Pi, \MA', \MO''}(\kappa, 2, c)$ differ, up to relabeling of the parties, only in one way : some $\Sign$ executions in $\Gm^\PSI_{\Pi, \MA, \MO'}(\kappa, n, c)$ have been replaced by $\Simulate$ in $\Gm^\PSI_{\Pi, \MA', \MO''}(\kappa, 2, c)$ and vice versa. Suppose $i \in \N$ such replacements have been made, then for $0 \leq j \leq i$ let $\Gm_j(\kappa, c)$ be $\Gm^\PSI_{\Pi, \MA, \MO'}(\kappa, n, c)$ with only the first $j$ such replacements made, i.e. $\Gm_0(\kappa, c) = \Gm^\PSI_{\Pi, \MA, \MO'}(\kappa, n, c)$ and $\Gm_i(\kappa, c) = \Gm^\PSI_{\Pi, \MA', \MO''}(\kappa, 2, c)$. Since, by construction, $\Prr{\Gm_0(\kappa, c) = c} - \frac{1}{2} \not\leq \negl(\kappa)$ and $\Prr{\Gm_i(\kappa, c) = c} - \frac{1}{2} \leq \negl(\kappa)$, we can fix a lowest $k$ such that $\Prr{\Gm_k(\kappa, c) = c} - \frac{1}{2} \not\leq \negl(\kappa)$ and $\Prr{\Gm_{k+1}(\kappa, c) = c} - \frac{1}{2} \leq \negl(\kappa)$. $\Gm_k$ and $\Gm_{k+1}$ differ only in one replacement. Without loss of generality, assume one $\Sign[P_s][P_v](m)$ was replaced by $\Simulate[P_s][P_v](m)$
    
    Now define an adversary $\MB$ for $\Gm^\NT$ as follows: $\MB(1,\params, \pk_S, \sk_S, \pk_V, \sk_V)$ picks a $c \in \{0,1\}$ and runs $\Gm_k(c, \kappa)$, replacing $\pk_s$ with $\pk_S$, $\sk_s$ with $\sk_S$, $\pk_v$ with $\pk_V$, and $\sk_v$ with $\sk_V$. This replacement is only a relabeling. The execution of $\Gm_k$ is stopped at the one difference with $\Gm_{k+1}$, then outputs $(m, (\state, c))$, where $\state$ is the current state of $\Gm_k$ and $m$ the message in the replaced $\Sign$. $\MB(2, (\state, c), \sigma)$ then continues the execution of $\Gm_k$ with $\sigma$ as the result of the replaced $\Sign$ until $\Gm_k$ outputs $c'$. $\MB$ then outputs $0$ if $c = c'$ and $1$ otherwise.
    
    Note that in $\Gm^\NT_{\Pi, \MB}(0, \kappa)$, i.e. the case where a $\Sign$ is used in the non-trans\-fer\-a\-bi\-li\-ty game, $\MB$ plays $\Gm_k(c, \kappa)$ and in $\Gm^\NT_{\Pi, \MB}(1, \kappa)$, $\MB$ plays $\Gm_{k+1}(c, \kappa)$. Thus we have that
    \[ \Prr[b]{\Gm^\NT_{\Pi, \MB}(b, \kappa) = b} = \frac{1}{2}\Prr[c]{\Gm_k(c, \kappa) = c} + \frac{1}{2}\Prr[c]{\Gm_{k+1}(c, \kappa) \neq c}.\]
    This directly implies that
    \[ \Adv^\NT_{\Pi, \MB}(\kappa, n) = \frac{1}{2}\left(\Prr[c]{\Gm_k(c, \kappa) = c} - \Prr[c]{\Gm_{k+1}(c, \kappa) = c}\right) \not\leq \negl(\kappa).\]
    
    This contradicts our assumption that $\Pi$ is computationally non-transferable, thus $\Pi$ must be $n$-party sender private.
\end{proof}

\section{Conclusion}

In this paper, we provided a way of defining sender privacy in the $n$-party setting that is novel for \DVS schemes, a generalization of existing definitions and in line with definitions for other types of schemes in the multi-party setting, in particular ID-based \SDVS schemes. We explored the effects of choosing the challenge differently and observed that this induces only polynomial differences in the advantage the adversary has. Furthermore, we showed how other properties of a \SDVS scheme can be used to boost the sender privacy of a scheme from an alternative definition to our definition. In particular, we have proven that under the assumption of strong unforgeability and computational non-transferability a 2-party sender-private scheme is $n$-party sender private. The proven relations are important since the \SDVS schemes are often meant to be employed in an n-party setting and we give sufficient conditions for this to be secure.

We would like to stress that the objective of this paper is to formulate sender privacy in such a way that it covers all theoretical types of attacks that should be intuitively covered by this property. Thus, the definition presented is not necessarily technically different from previous definitions, in fact, it will coincide in many cases. As such we do not provide separating examples of schemes that satisfy one definition but not another, as any such case would be extremely artificial. Instead, the definitions in this work and their equivalence should be used to simplify proofs where sender privacy property is used, both in classical and quantum use cases.

\section{Acknowledgements}

JvW is supported by the Luxembourg National Research Fund (FNR), under the joint CORE project Q-CoDe (CORE17/IS/11689058/Q-CoDe/Ryan).

\addcontentsline{toc}{section}{References}
\emergencystretch=1em
\printbibliography
\emergencystretch=0pt

\end{document}

%% file: commands.tex
\usepackage{xifthen}
\usepackage{twoopt}

\DeclareMathOperator{\Tr}{Tr}

\DeclareMathOperator{\negl}{negl}

\newcommand{\eps}{\varepsilon}

\newcommand{\Trr}[2][]{\Tr\ifthenelse{\isempty{#1}}{}{_{#1}}\left[#2\right]}
\newcommand{\Prr}[2][]{\Pr\ifthenelse{\isempty{#1}}{}{_{#1}}\left[#2\right]}

\newcommand{\axx}[1]{\expandafter\newcommand\csname#1\endcsname{\axname{#1}\xspace}}
\newcommand{\saxx}[2]{\expandafter\newcommand\csname#1\endcsname{\axname{#2}\xspace}}
\axx{AKG}
\axx{AUX}
\axx{View}
\axx{SIM}
\axx{PPT}
\axx{QKE}
\axx{QKD}
\axx{QPT}
\axx{AuthBB}
\axx{SendQ}
\saxx{QTC}{Q2C}
\saxx{CTQ}{C2Q}
\axx{start}
\axx{initiator}
\axx{responder}
\axx{poly}
\axx{Adv}
\axx{Den}
\saxx{Gm}{G}
\axx{params}
\axx{state}
\axx{Setup}
\axx{SUSDVS}
\axx{SDVS}
\axx{DVS}
\axx{NT}
\axx{UF}
\saxx{PSI}{SendPriv}
\axx{Hua}
\saxx{ADVPSI}{ChosenSendPriv}
\axx{msg}
\axx{pid}
\axx{Partner}
\axx{SIS}
\axx{ISIS}
\axx{JSI}
\axx{NJ}
\axx{SKM}
\newcommand{\SignL}[3]{\Sign(\sk_{#1}, \pk_{#1}, \pk_{#2}, #3)}

\newcommand{\SimulateL}[3]{\Simulate(\sk_{#2}, \pk_{#2}, \pk_{#1}, #3)}
\newcommandtwoopt{\Sign}[2][][]{\axname{Sign}{\ifthenelse{\isempty{#1}}{}{_{#1\rightarrow #2}}}\xspace}
\newcommandtwoopt{\Verify}[2][][]{\axname{Verify}{\ifthenelse{\isempty{#1}}{}{_{#1\rightarrow #2}}}\xspace}
\newcommandtwoopt{\Simulate}[2][][]{\axname{Simulate}{\ifthenelse{\isempty{#1}}{}{_{#1\rightarrow #2}}}\xspace}

\newcommand{\sk}{\ensuremath{\mathsf{sk}}\xspace}
\newcommand{\pk}{\ensuremath{\mathsf{pk}}\xspace}

\newcommand{\KeyGen}{\axname{KeyGen}}

\newcommand{\MA}{\mathcal{A}}
\newcommand{\MB}{\mathcal{B}}
\newcommand{\CC}{\mathcal{C}}

\newcommand{\MP}{\mathcal{P}}
\newcommand{\MO}{\mathcal{O}}

\newcommand{\ch}[2][]{{\color{gray}\ifthenelse{\isempty{#1}}{#2}{#1\rightarrow #2}}}

\newcommand{\N}{\mathbb{N}}

\newcommand{\axname}[1]{\ensuremath{\mathsf{#1}}}

\definecolor{dblue}{rgb}{0,0,.8}
\definecolor{dred}{rgb}{.7,.05,.05}
\definecolor{dgreen}{rgb}{.1,.5,.1}